\documentclass[12pt]{article}
\usepackage[utf8]{inputenc}
\usepackage{graphicx}

\title{Dempster-Shafer P-values: Thoughts on an Alternative Approach for Multinomial Inference}
\author{Kentaro Hoffman,  Kai Zhang, Tyler McCormick, Jan Hannig}
\usepackage{amsmath}
\usepackage{amssymb}
\usepackage{amsfonts}
\usepackage{longtable}
\usepackage[english]{babel}
\usepackage{amsthm}
\usepackage{amsfonts} 
\newtheorem{theorem}{Theorem}

\newtheorem{definition}{Definition}[section]
\usepackage[dvipsnames]{xcolor}
\usepackage{float}
\usepackage{hyperref} 
\usepackage{bm}
\usepackage{hyperref}
\hypersetup{colorlinks,urlcolor=blue}
\usepackage{svg}

\usepackage[square,numbers]{natbib}
\usepackage{adjustbox}
\usepackage{array}

\newcolumntype{R}[2]{%
    >{\adjustbox{angle=#1,lap=\width-(#2)}\bgroup}%
    l%
    <{\egroup}%
}
\newcommand*\rot{\multicolumn{1}{R{45}{1em}}}

\newtheorem{exmp}{Example}[section]

\begin{document}

\def\spacingset#1{\renewcommand{\baselinestretch}%
{#1}\small\normalsize} \spacingset{1}

\maketitle
In this paper, we demonstrate that a new measure of evidence we developed called the Dempster-Shafer p-value which allow for insights and interpretations which retain most of the structure of the p-value while covering for some of the disadvantages that traditional p-values face. Moreover, we show through classical large-sample bounds and simulations that there exists a close connection between our form of DS hypothesis testing and the classical frequentist testing paradigm. We also demonstrate how our approach gives unique insights into the dimensionality of a hypothesis test, as well as models the effects of adversarial attacks on multinomial data. Finally, we demonstrate how these insights can be used to analyze text data for public health through an analysis of the Population Health Metrics Research Consortium dataset for verbal autopsies.

\section{Introduction}
Although p-values have historically played a ubiquitous role in many statistical analyses, the past decade has seen a flurry of criticism and scandals on their use and misuse. For example, in a highly cited \textit{Nature} article, the authors found in a meta-analysis across 4 highly cited sociological and medical journals, an incredible 402 out of 791 $(51\%)$ articles erroneously conflate non-significance with no effect \cite{amrhein_greenland_mcshane_2019}, a basic misconception commonly covered in a variety of introductory statistics textbooks \cite{rowntree_2018,illowsky_dean_illowsky_2018,diez_barr_cetinkaya-rundel_2019} In another, in response to the landmark 2016 \textit{American Statistical Association} statement of propose on the proper use of p-values, the \textit{New England Journal of Medicine} announced a change in reporting standards for p-values \cite{harrington_al}. However, this new standard itself misinterpreted a p-value that claims ``A p-value of 0.05 carries a 5\% risk of a false positive result" \cite{stone_2016}. This was despite the ASA statement of propose explicitly stating that this is not correct. 

In response to these issues, the past few years have seen a new movement in statistics- one that seeks to end the uncritical use and misuse of the 0.05 p-value. This paper seeks to join this and demonstrate that a new measure of evidence we developed called ``Dempster-Shafer p-value" provides new insights and interpretations which retain most of the structure of the p-value while covering for some of the disadvantages that traditional p-values face. While this is focused on p-values for multinomial inference, it is the author's hope that the useful properties, advantages, and disadvantages that arise from incorporating Dempster-Shafer inference to the p-value will be thought-provoking in the continued conversation on the future of statistical hypothesis testing.  Finally, we demonstrate the utility of the proposed approach in the context of verbal autopsies, a strategy commonly used in global health to determine the distribution of deaths by cause in settings where deaths occur outside of hospitals and are not formally recorded.

\section{Demspter-Shafer Inference}

Dempster-Shafer (DS) Inference was developed in the 1960-70s by Arthur Dempster \cite{10.1214/aoms/1177698328} and Glen Shafer \cite{dempster1976mathematical} as an alternative approach to prior-free statistical inference. What makes DS distinctive is that for an event $A$, it ascribes the triple, ($p(A)$, $p(A^c)$, $1 -p(A) - p(A^c))$, to quantify the evidence \textbf{for $A$}, \textbf{against $A$}, and \textbf{unknown} rather than traditional inferential approaches that denote a probability double $(P(A), 1-P(A))$.  Historically, DS Inference has played an important role in the development of artificial intelligence \cite{yager2008classic} and reliability engineering \cite{sentz}, leading to the development of the field of imprecise probability. 

Despite its popularity in other fields, DS Inference has had a smaller impact on the statistical world in which it originated. It has been hypothesized that it failed to take off due to technical difficulties in interpreting beliefs and plausibilities \cite{PEARL1988211}, computational burden, and perhaps the biggest problem: the lack of conventional long-run frequentist properties under repeated sampling \cite{10.1214/10-STS322}. However, in the last several years, there has been renewed interest from the statistical community in other forms of prior-free forms of posterior inference which incorporate epistemic uncertainty. This has come from closely related topics such as  Generalized Fiducial Inference \cite{gfi}, Confidence Distributions \cite{https://doi.org/10.1111/insr.12000},
Inferential Models \cite{martin2015inferential}, and faster computational techniques using Gibbs sampling on polytopes  \cite{doi:10.1080/01621459.2021.1881523}. We build on this rising wave and demonstrate that a Dempster-Shafer inference procedure known as \textit{Dirichlet-DS} can be adapted to perform a wide variety of Dempster-Shafer Hypothesis Tests which are computationally fast and model epistemic uncertainty: in particular, we focus on the uncertainty that can result from adversarial attacks. Moreover, we will show through classical large sample bounds and simulations that there exists a close connection between our form of DS hypothesis testing and the classical frequentist testing paradigm. Finally, we will demonstrate how our approach gives unique insights into the dimensionality of a hypothesis test, as well as models the effects of adversarial attacks on multinomial data.

\section{Multinomial Data Generation via \textit{Dirichlet-DSM}}
In a Multinomial Data generation scheme, we observe $n$, $k$-dimensional
multinomial observations $(n_1, ... n_k)$ with unknown but fixed probabilities $\mathcal{P} = (p_1, ... p_k)$ such that $\sum_{i=1}^k p_i = 1$
$$n_1, ... , n_k \sim \text{ Multinomial }(1, p_1, ... , p_k). $$
Each draw from the multinomial can be represented in component form as a sum of binary vectors using $y_{ij} \in \{0,1\}$ for $i \in \{1, ... n\}$ and $j \in \{1, ... , k\}$ such that $\sum_{j} y_{ij} = 1$ and $\sum_{i} y_{ij} = n_j$. We will interchangeably referring to count and component form as necessary.\\\\
To perform DS inference, it is first necessary to define the data generation process that yielded $(n_1, ..., n_k)$. This can be done through prior insight into the data generation scheme, or it can be chosen to facilitate the inference of the parameters. Consistent with the second reason, we will employ the data generation process described in \citet{lawrence} as the \textit{Dirichlet-DSM} process. We first define an unknown permutation $\Pi$ of the $k$ classes as: $\Pi(\{1,...,k\}) = \{ \pi(1), ..., \pi(k)\}$ and the inverse permutation as $\Pi^{-1}(\{ \pi(1), ..., \pi(k)\}) = \{1, ... , k\}$. This allows us to define $I_{\pi,a}$, as an interval of length $p_a$ whose location on the unit interval is determined by:
$$I_{\pi,a} = \begin{cases}
[0, p_a) & \text{ if } \pi(1) = a \\
[\sum_{i=1}^{\Tilde{a}-1} p_{\pi(i)} , \sum_{i=1}^{\Tilde{a}-1} p_{\pi(i)}  + p_a) & \text{ if } \pi(1) \neq a 
\end{cases} $$
where $\tilde{a} = \pi^{-1}(a)$ for $a \in \{1, ..., j\}$. Regardless of choice of permutation $\Pi$, each interval $I_{\pi,a}$ will be of length $p_a$ and combing together all the intervals yields $[0,1)$. Thus we can generate multinomial draws through the data generating process:
$$y_{ij} = a  \textit{   iff  }   U_i \in I_a  $$
where $(U_1, ... U_n)$ are $n$ independent, uniform random variables on [0,1]. This data generating process holds two distinct advantages over related processes such as \textit{ Simplex-DSM}  \citet{10.1214/aoms/1177699517},  \textit{interval-DSM}  in \citet{gfi}, and \textit{Simplex-DS} in \citet{doi:10.1080/01621459.2021.1881523}. First, unlike \textit{interval-DSM}, this inference procedure for \textit{Dirichlet-DSM} is not sensitive to the order of the categories. Further, unlike the \textit{simplex-DSM} and \textit{Simplex-DS}, \textit{Dirichlet-DSM} has a relatively simple expression for its posterior estimation which requires no expensive acceptance-rejection or Gibbs sampling, this makes the method more feasible when combine with the subsequent repeated required for hypothesis testing.
\subsection{A Recipe for Multinomial DS Inference using \textit{Dirichlet-DSM}}
Inferring $\mathcal{P}$ requires coming up with an inversion of the data generating equation. We first assume that we have observed $n$ multinomial draws generated by the \textit{Dirichlet-DSM} scheme and the permutation that was used for the draws, $\Pi$. For each multinomial draw $(y_{i1}, ... , y_{ij})$, let $a_i \in \{1, ... , k\}$ be the observed component, $a_i = \sum_{j=1}^k j I(y_{ij} = 1)$ for the $i$-th draw. Then, define $\tilde{\Pi}_{i, \Pi, a_i} (\{1, ..., k \}) = (\tilde{\pi}(1), ... \tilde{\pi}(k))$ as any permutation that has the observed category, $a_i$, first in its order. In other words, this means that $\tilde{\pi}(a_i) = 1$. Now, let $D_{\tilde{\Pi}_i, a_i}$ be a function which takes a point on $[0,1)$, and maps it back to $[0,1)$ but with the order of the intervals permuted by $\tilde{\Pi}$. Being a well defined bijection determine by $\tilde{\Pi}$, $D$ also admits a well defined inverse $D^{-1}_{\tilde{\Pi}_i, a_i}$ which undoes the reordering. By applying $D^{-1}_{\tilde{\Pi}_i, a_i}$ to $U_i$, we can define an auxiliary variable:
$$W_i = D_{\tilde{\Pi}_i, a_i}(U_i).  $$
Due to the bijective nature of $D_{\tilde{\Pi}_i, a_i}$, $W_i$ is distributed $U[0,1)$. When conditioning $W_i$ on the observed data, $a_i$, $W_i$ becomes:
$$W_i | a_i \sim U(0, p_{a_i}). $$
To perform a fiducial inversion, we start with an iid copy of $W_i$ which we call $W_i^*$. Being iid, $W_i^*$ is also distributed $U(0,1)$ and we would like to derive the distribution of $W_i^*$  conditional on this being consistent with the observed data (This is refers to the "Continuing to believe" justification in DS inference). When $U_i^*$ is compatible with $a_i$, this implies that $U_i^* \in I_{{\Pi}_{i}, a_i}$ and by construction $U_i^* \in I_{\tilde{\Pi}_{i}, a_i}$. Applying $D$'s bijective behavior, this yields the equivalent condition that: 
$$W_i ^* \leq p_{a_i} .$$
Thus, our posterior random set is defined by jointly generating $n$ auxiliary variables $W_i ^*$ such that each satisfies:
$$W_i ^* |  \sum_i  W_i^*\leq 1.$$
As for the exact definition of this joint distribution, this constraint on the auxiliary variables is identical to the constraints used in theorem 4.1 of \citet{lawrence} so both though DS and Fiducial derivations, we arrive to the same statement that:
\begin{theorem}
    Given data from a $k$ dimensional multinomial, $(n_1, ... , n_k)$ the posterior random set $\hat{\mathcal{P}}$ is distributed according to:
    \begin{align*}
        \hat{\mathcal{P}} =(\hat{p}_1, ... , \hat{p}_k) &: \hat{p}_a \geq  Z_a, a \in \{1, ... ,k\} \\
        & (Z_0, Z_1, ... , Z_k) \sim Dirichlet(1, n_1, ..., n_k)
    \end{align*}
    \label{eq:prs1}
\end{theorem}
For those that prefer a more geometric perspective, one can think of $\hat{\mathcal{P}}$ as a polytope, specifically a simplex, with edges at:  
$$\bm{v}_j = (Z_1, ... Z_{j-1} , Z_{j} + Z_0, Z_{j+1}, ... Z_k) ,\,  \forall j \in \{1, ... k \}. $$
Any element in $\hat{\mathcal{P}}$ can be written as:
$$\mathcal{\hat{P}}(z) = \{z_{-0} + \bm{\theta}^T z_0 |\bm{\theta} = (\theta_1, ... , \theta_k) ,\,\sum_{j=1}^k \theta_j = 1 ,\, 0 \leq  \theta_j \leq 1 ,\, \text{ for } j = 1, ... ,k  \}$$
or in vector form:
$$\mathcal{\hat{P}}(z)  =  \begin{pmatrix}
Z_1 + \theta_1 Z_0 \\
\vdots \\
Z_k + \theta_k Z_0
\end{pmatrix}. $$

As the edges of the simplex is defined by a random variable, this of course makes the location of the simplex random a random variable. However, the behavior of this randomness is not without good properties as one can show that as desired, the location of our posterior random set converges to the true probabilities $\mathcal{P}$ in probability. 
\begin{theorem}
$\mathcal{\hat{P}}(\mathcal{Z})$ converges in probability to $\mathcal{P}$.
\label{thrm:1}
\end{theorem}

\begin{proof}
    It suffices to show that the edges of the simplex converge to $\mathcal{P}$. By construction of $\bm{v}_j$, this occurs when $Z_0 \rightarrow^{P} 0$ and $Z_i \rightarrow^{P} p_i$ for $i \in \{1, ..., k\}$. 
    \begin{align*}
        \lim_{n \rightarrow \infty} P(|(Z_0,Z_1, ..., Z_k) - (0,p_1, ..., p_k)|) & \leq  \lim_{n \to \infty} P(|Z_0| \geq \frac{\epsilon}{2k} ) +  \sum_{i=1}^k  \lim_{n \to \infty} P(|Z_0 - p_i| \geq \frac{\epsilon}{2k} ) \\
        & \leq \lim_{n \to \infty} \frac{2k(1-1/n)}{n ( 1+n) \epsilon} + \sum_{i=1}^k \lim_{n \to \infty} \frac{\frac{n_i}{n} (1-  \frac{n_i}{n})}{n+1}  \frac{2 k }{\epsilon}  \\
        & = \sum_{i=1}^k \lim_{n \to \infty} \frac{\frac{n_i}{n} (1-  \frac{n_i}{n})}{n+1}  \frac{2 k }{\epsilon} \\
        & = \sum_{i=1}^k \lim_{n \to \infty} \frac{p_i (1-  p_i)}{n+1}  \frac{2 k }{\epsilon} =0
    \end{align*}
\end{proof}

\section{Point Estimation with DS Inference}
Before we describe hypothesis testing, we must describe how we will be turing our random polytopes $\mathcal{\hat{P}}(z)$, into point estimates. Recall that $\mathcal{\hat{P}}(z)$ is a polytope that describes the region of feasible parameter estimates. To emphasize this fact, we will denote $\mathcal{\hat{P}}(z)$ using $\Delta(z)$ and $\mathcal{\hat{P}}(Z)$ using $\Delta(Z)$. Using a test statistic $T(\mathcal{\hat{P}}, \mathcal{P}_0)$ and fixed $(n_1, ... , n_k)$, one can compute an \textbf{upper},  \textbf{mean}, and \textbf{lower test statistic} for each $\Delta(z)$ as follows:
\begin{align}
  T_{upper}(z) &=  \sup_{\mathcal{\hat{P}} \in \Delta(z) }  T(\mathcal{\hat{P}}, \mathcal{P}_{Observed}) \\
  T_{mean}(z)  &=   T(E_{ \Delta(z) } \mathcal{\hat{P}},\mathcal{P}_{Observed}) \\
  T_{lower}(z) &=  \inf_{\mathcal{\hat{P}} \in \Delta(z) } T(\mathcal{\hat{P}}, \mathcal{P}_{Observed}).
\end{align}
$E_{\Delta(z)}$ in this case represents taking the expectation with respect to the uniform measure upon $\Delta(z)$. 
\\\\
By plugging in the random variable $\mathcal{Z}$, this turns the upper, mean and lower statistics into \textbf{upper}, \textbf{mean}, and \textbf{lower sampling distributions}:
\begin{align}
  T_{upper} &=   \sup_{\mathcal{\hat{P}} \in \Delta(\mathcal{Z}) }  T(\mathcal{\hat{P}}, \mathcal{P}_0)  \\
  T_{mean}   &=   T(E_{ \Delta(\mathcal{Z}) } \mathcal{\hat{P}}, \mathcal{P}_0)   \\
  T_{lower}  &=  \inf_{\mathcal{\hat{P}} \in \Delta(\mathcal{Z}) }  T(\mathcal{\hat{P}}, \mathcal{P}_0).
\end{align}
Assuming that $T$ defines a valid distance metric, $T_{upper}(z)$ represents the largest possible test statistic one could create on $\Delta(z)$, $T_{lower}(z)$ represents the smallest, and $T_{mean}$ represents the average case. 

\subsection{Comparison to other Frequentist point estimation}
The creation of the upper, mean, and lower test statistics contrasts with the frequentist hypothesis testing paradigm. Under a frequentist paradigm, one derives a point estimate ,$\mathcal{\hat{P}}_{Freq}$,  using the maximum likelihood equation, and evaluates the asymptotic distribution of the point estimate plugged into the test statistic:
$$T_{Freq} = T(\mathcal{\hat{P}}_{Freq}, \mathcal{P}_0) $$
where $\mathcal{P}_0$ is the parameter value for the null hypothesis. In the case of a perason chi-squared hypothesis test, $\hat{\mathcal{P}}$ would be the  estimated proportions, ${\mathcal{P}}_0$ would be the proportions under the null hypothesis, and $T = \sum_i \frac{n_i \hat{p}_i - n_i p_{0i} }{ n_i p_{0i} } $. It turns out, one can connect the frequentist test statistic with that of the DS through the choice of summary statistic used on the posterior random set.

\begin{theorem}
For fixed observations $(n_1, ... , n_k)$ from a multinomial distribution and test statistic T, $T_{mean}$ is equal to $T_{Freq} = T(\hat{\mathcal{P}}_{Freq}, \mathcal{P}_0)$ where $\mathcal{\hat{P}}_{Freq} = (\frac{n_1 +1/k}{n+k}, ... , \frac{n_k +1/k}{n+k})$
\end{theorem}

\begin{proof}
Since:
\begin{align*}
    T_{mean} &=  T(E_{\Delta(Z)}\hat{\mathcal{P}}, \mathcal{P}_0) \\
    T_{Freq} &= T(\hat{\mathcal{P}}_{Freq}, \mathcal{P}_0)
\end{align*}
and $\mathcal{P}_0$ is fixed, it suffices to show that:
$$ E_{\Delta(Z)} \hat{\mathcal{P}} = \hat{P}_{Freq}.$$
For each component of $i \in \{1, ... k \}$, note that we can write:
\begin{align*}
    E_{\Delta(Z)}(\hat{p}_i) &= E_{Z} E_{\Delta(z)}[ (\hat{p}_i)] \\
    &=E_{Z} E_{\Delta(z)} [ z_{-0} + \theta_i z_0] \\
    &=E_{Z}(z_{-0}) + \frac{1}{k} E_{Z}(z_0) \\
    &= \frac{n_i + 1/k}{n+1}
\end{align*}
\end{proof}
Note that at this point, $T_{mean}$ is the Bayes estimator under MSE of when the $(p_1, ... , p_k)$ has a prior distribution of Dirichlet($1/k, ... , 1/k)$ \cite{LehmCase98}. This illustrates that the frequentist approach of choosing an arbitrary point estimate is equivalent to choosing an arbitrary way to summarize the posterior random sets $\mathcal{\hat{P}}(Z)$. To the best of our knowledge, we are unaware of other results which connect DS to frequentist via the choice of summmarization.

\section{Hypothesis Testing with DS Inference}
\label{sec:hypo}
With our upper, lower, and mean point estimates and their distributions defined, we are ready to define how we perform a level $\alpha$-hypothesis test with our DS approach. 

\begin{definition}{Level $\alpha$ two-sided DS Hypothesis Test}\\
Let us have two competing hypotheses $H_0$ and $H_1$ about $k$-dimensional parameter $\mathcal{P}_0 \in \Delta^{k-1}$  such that:
$$H_0:\mathcal{P} = \mathcal{P}_0, H_1: \mathcal{P} \not = \mathcal{P}_0.$$
Using $\mathcal{P}_0$ and a point estimate, $\mathcal{\hat{P}}_{observed}$ based on fixed counts $(n_1, ... , n_k)$, we compute tail probabilities for the upper and lower distributions:
\begin{align}
  \pi_{upper} = P(T_{lower} &\geq T(\hat{P}_{observed},\mathcal{P}_0)) \\
  \pi_{lower} = P(T_{upper} &\geq T(\hat{P}_{observed},\mathcal{P}_0)). 
\end{align}
We then conclude based on the tail probabilities:\\
$$
\begin{cases}
\text{Reject } H_0 &  \text{ If } \pi_{upper} \leq \alpha  \\
\text{Accept } H_0 &  \text{ If } \pi_{lower} > \alpha  \\
\text{Unknown}  &  \text{ If } \pi_{lower} \leq \alpha \text{ but }  \pi_{upper} > \alpha \\
\end{cases}
$$
\end{definition}

\begin{exmp}
As an example, let us perform a DS goodness-of-fit test using the chi-squared test statistic. The data comes from a Multinomial (n, $(\frac{1}{3}, \frac{1}{3}, \frac{1}{3})$ and the null hypothesis is that all three classes are equally likely. The number of elements observed in each class is $(n_1 ,n_2,n_3) =(3,2,5)$. In such a case, our upper and lower test statistics are based on the chi-squared statistic:
$$T_{upper}(\mathcal{Z}) = \sup_{(\hat{p}_1, \hat{p}_2, \hat{p}_3) \in \Delta(\mathcal{Z})}\sum_{i=1}^3 \frac{(10\hat{p}_i - 10/3 )^2}{10/3}$$
$$T_{lower}(\mathcal{Z}) = \inf_{(\hat{p}_1, \hat{p}_2, \hat{p}_3) \in \Delta(\mathcal{Z})}\sum_{i=1}^3 \frac{(10\hat{p}_i - 10/3)^2}{10/3}$$
where elements of $\Delta(\mathcal{Z})$ are:
$$\mathcal{\hat{P} ({\mathcal{Z})}} = \begin{pmatrix}
\hat{p}_1 \\
\hat{p}_2 \\
\hat{p}_3 \\
\end{pmatrix} = \begin{pmatrix}
Z_1 + \theta_1 Z_0 \\
Z_2 + \theta_2 Z_0 \\
Z_3 + \theta_3 Z_0
\end{pmatrix}   $$
for fixed $\theta_1 + \theta_2 + \theta_3 = 1$ and $0 \leq \theta_i \leq 1$ for $i \in \{1,2,3\}$ and:
$$(Z_0,Z_1,Z_2, Z_3) \sim Dirichlet(1,3,2,5).  $$
In Figure \ref{fig:fig1}, we plot out 100 realizations of posterior random sets as well as the the null point, $\mathcal{P}_0 = (\frac{1}{3}, \frac{1}{3}, \frac{1}{3})$ and point estimate $(\frac{n_1 + 1}{n+3}, \frac{n_2 + 1}{n+3}, \frac{n_3 + 1}{n+3}) = (\frac{4}{13}, \frac{3}{13}, \frac{6}{13})$. Note how most of the polytopes are clustered around the point estimate which is some distance from the null point.

\begin{figure}[H]
    \centering
    \includegraphics[scale = 0.15]{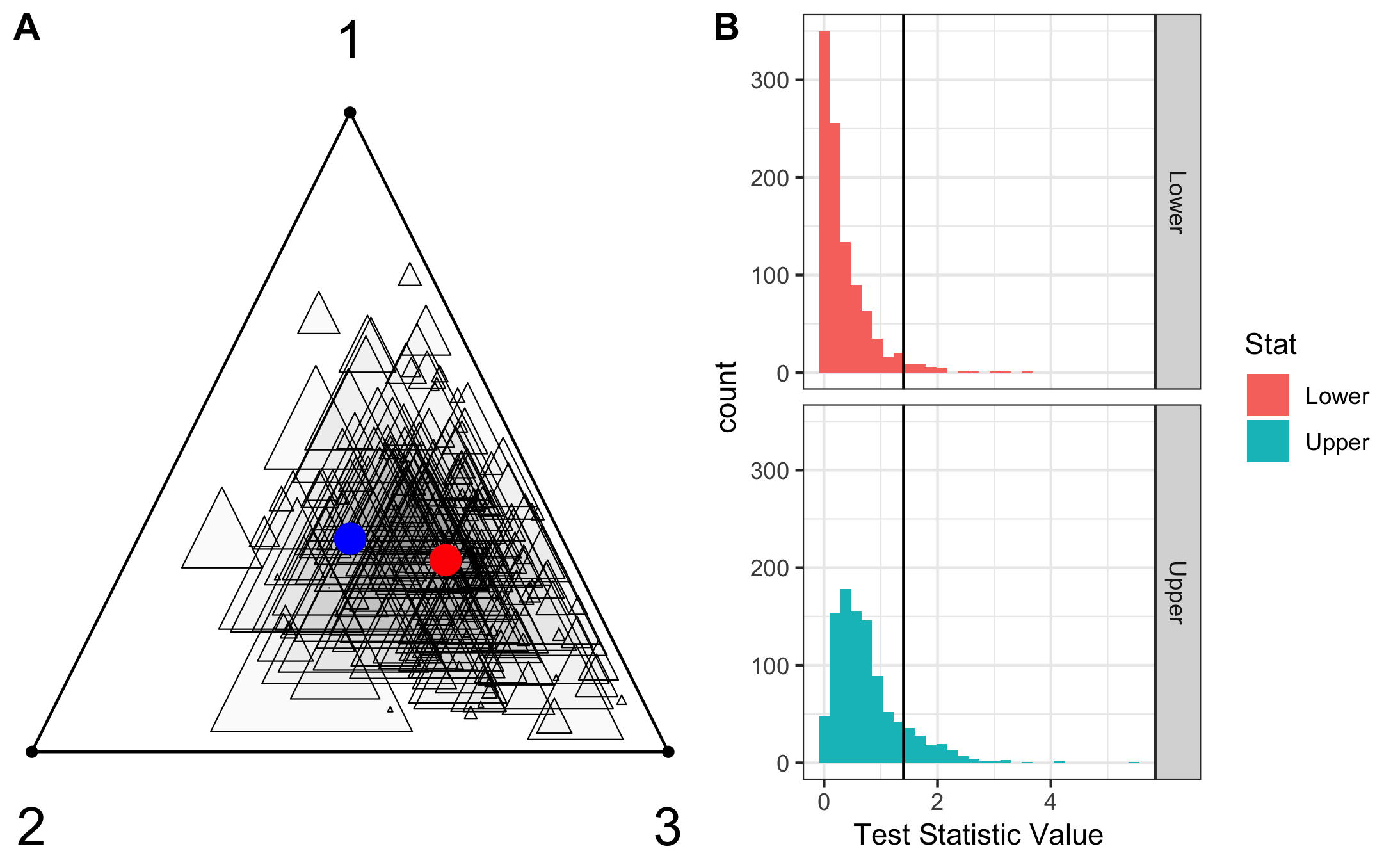}
    \caption[Dempster Polytopes for n = 10]{\textbf{A)} 100 randomly chosen posterior random sets from a 3 dimensional test of uniformity. Note how the polytopes are centered around the point estimate (red). 100 was chosen for visibility reasons. \textbf{B)} 1000 simulations of the upper and lower test statistic and a vertical line representing $\mathcal{P}_0$. For this data,  $\pi_{lower}$ is 0.12 and $\pi_{upper}$ is 0.034.  }
    \label{fig:fig1}
\end{figure}
In the Figure \ref{fig:fig2}, we have summarized 1000 posterior random sets into their upper and lower chi-squared test distributions. As is, the lower tail probability is 0.034 while the upper tail probability is 0.12, so at the $\alpha = 0.05$ level, we conclude that we lack the power to make a conclusion either way. More precisely, if we performed a hypothesis test using our most optimistic upper test statistic, we would reject while our conservative lower test statistic would fail to reject. Having posterior estimates leading to opposite conclusions is not a sign of reliability so we conclude that we cannot make a conclusion either way. 
\\\\
Now, if we increase our sample size from n=10 to n=100 and observe $(n_1, n_2, n_3) = (30,20,50)$, we now have much more evidence that the null hypothesis of each class being equally likely is incorrect. Note that our posterior random sets are now defined by:
$$(Z_0,Z_1,Z_2, Z_3) \sim Dirichlet(1,30,20,50). $$
The width of each random set is $Z_0 \sim Beta(1, 100)$, rather than $Z_0 \sim Beta(1, 10)$ as in the previous example. Correspondingly, in Figure \ref{fig:fig2}, we can see in part A) that the polytopes are smaller and more clustered around our point estimate. This corresponds with our intuition that a larger sample size should result in more concentrated posterior estimates. Now none of the 1000 posterior random sets are close to the null point. Thus, $\pi_{lower}$ and $\pi_{upper}$ are both $<0.001$ meaning that we reject our null hypothesis with confidence.
\begin{figure}[H]
    \centering
    \includegraphics[scale = 0.15]{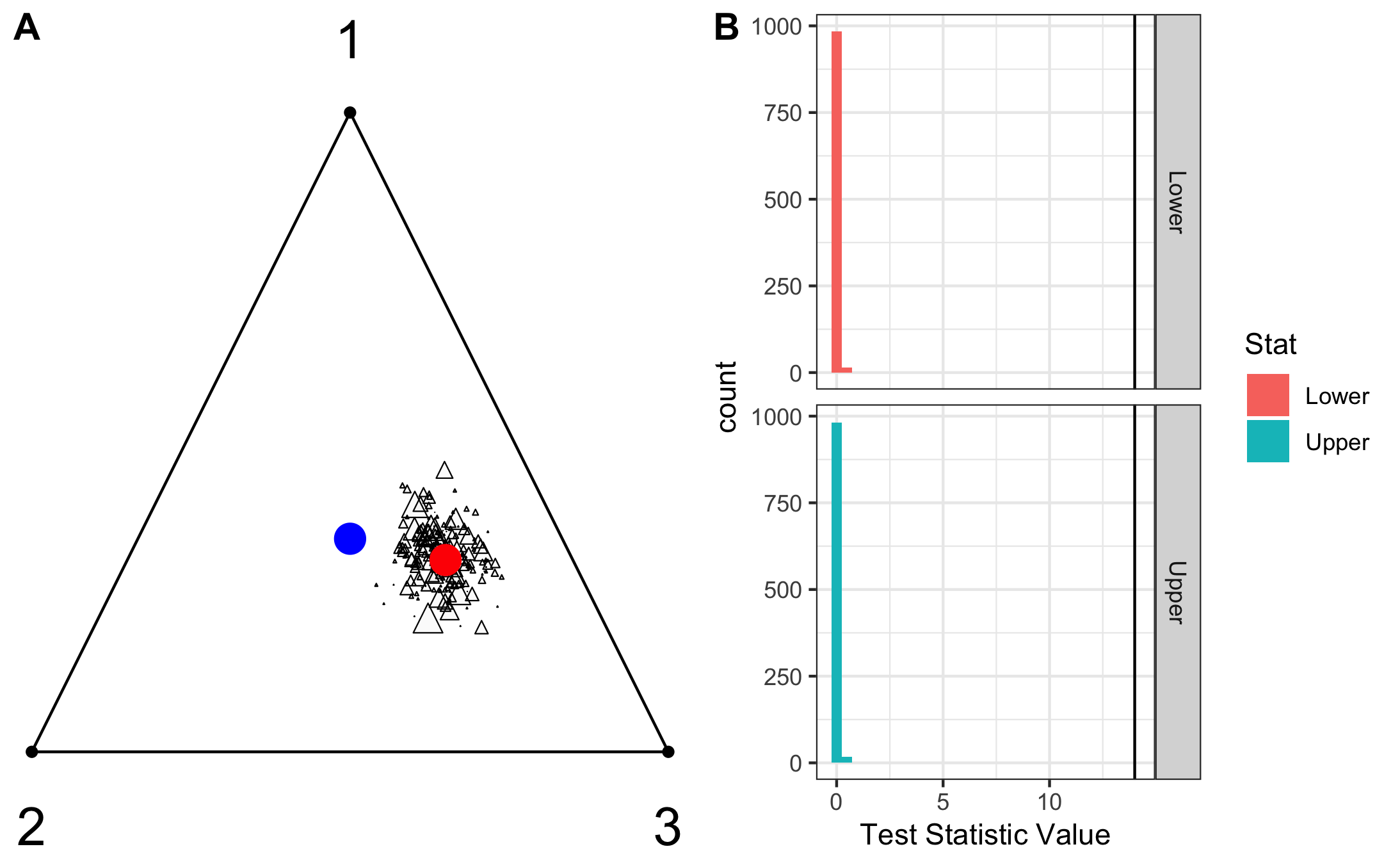}
    \caption[Dempster Polytopes for n = 100]{\textbf{A)} 100 randomly chosen posterior random sets. Note how the polytopes are more tightly centered around the point estimate than in Figure \ref{fig:fig1}.  \textbf{B)} 1000 simulations of the upper and lower test statistic and a vertical line representing $\mathcal{P}_0$. With the larger sample size, both $\pi_{lower}$ and $\pi_{upper}$ are $< 0.0001$.}
    \label{fig:fig2}
\end{figure}
As we saw in the previous example, increasing our sample size made us go from an uncertain conclusion to a certain one. This turns out to be not only a feature of the previous example, but holds for any continuous, bounded test statistic.

\end{exmp}

\begin{theorem}
Consider a k-dimensional DS hypothesis as described above based on a continuous and bounded test statistics, $T({\mathcal{\hat{P}}}, \mathcal{P}_0)$. As the sample size increases, the probability of a DS hypothesis test returning an unknown result goes to 0 in probability.
\end{theorem}
\begin{proof}
We first desire to show that $T_{upper} \overset{P}{\to} T(\mathcal{P},\mathcal{P}_0)$. This requires that:
\begin{align*}
    \lim_{n \to \infty} P(| \sup_{\hat{p} \in \Delta(\mathcal{Z})} T(|\hat{P}, \mathcal{P}_0) - T(\mathcal{P},\mathcal{P}_0| \geq \epsilon)  =0
\end{align*}
Since $\Delta(\mathcal{Z}) \sim Dirichlet(1, n_1, ..., n_k)$, we have by the strong law that $\Delta(\mathcal{Z})$ converges to a degenerate distribution centered at $\mathcal{P}$. If we let $A_n  = \{ \hat{p}: T(\hat{p},\mathcal{P}_0) \geq T(\mathcal{P},\mathcal{P}_0) \}$ $A_n$ converges point wise to $\mathcal{P}$. Thus, by monotone convergence, we get that:
\begin{align*}
& \lim_{n \to \infty} P(| \sup_{\hat{p} \in \Delta(\mathcal{Z})} T(|\hat{P}, \mathcal{P}_0) - T(\mathcal{P},\mathcal{P}_0)| \geq \epsilon) \\
    &= \lim_{n \to \infty} P(| \sup_{\hat{p} \in A_n} T(|\hat{P}, \mathcal{P}_0) - T(\mathcal{P},\mathcal{P}_0)| \geq \epsilon)  \\
    &=  P(| T(\mathcal{P}, \mathcal{P}_0) - T(\mathcal{P},\mathcal{P}_0)| \geq \epsilon) =0
\end{align*}
By the same argument, one can show that $T_{lower} \overset{P}{\to} T(\mathcal{P}, \mathcal{P}_0)  $.\\\\
To finish the proof, recall that we return an ``unknown" conclusion when $\pi_{upper} \leq \alpha$ but $\pi_{lower} > \alpha$. However since $T_{upper}$ and $T_{lower}$ are converging in probability to $T(\mathcal{P},\mathcal{P}_0)$, $\pi_{upper} \rightarrow \pi_{lower}$. And it becomes possible to return an "unknown" result.
\end{proof}

\section{Interpretations and Connections to Other Forms of Inference}
Now that we have defined what we claim to be a DS hypothesis test, we show how this test fits into the DS, Fiducial, and frequentist paradigm. 

\subsection{DS Interpretation}

To see how this procedure fits into the DS framework, we will first abstract this procedure in the spirit of Fraser's structural inference \cite{fraser} and illustrate how our tail probabilities define proper DS upper and lower probabilities. \\\\
First, let us write our data generation scheme in terms of a Fraser inspired \textit{randomized structural equation}:
$$\mathcal{X} = G(\mathcal{P}, \mathcal{U}|\pi). $$
here $G$ is our deterministic data generating equation, $\mathcal{P}$ are our parameters of interest, $\mathcal{U}$ is the random component and $\pi$ is an unknown permutation. Unlike Fraser's structural equation framework, we will not be assuming a group structure and will allow for aspects of the data generating equation to contain randomized components via $\pi$. We will further expand upon the implications of this randomized data generating equation in the Fiducial Inference interpretation section. \\\\
Using this, we can define the subset of dependent uniform random variables $\mathcal{U} = (U_1, ... U_n) $ that could have generated $(n_1, ... n_n)$ for some $\mathcal{P} = (p_1, ... p_k)$ as:
$$\mathcal{R}_x = \{ \mathcal{U} \in [0,1]^n : \exists \mathcal{P}  \text{ and } \exists \pi \in \mathcal{S}(\{1, ... ,k\}) \text{ st } (n_1, ... n_n) = G(\mathcal{P}, \mathcal{\textbf{u}}| \pi)   \} .$$
Here, $\mathcal{S}(\{1, ... ,k\})$ is 
the set of all permutations, $\pi$ on the set $\{1, ... ,k\}$ which is. Given $\textbf{u} \in \mathcal{R}_x$ , there exists a non-empty ``feasible" set on the parameter space  $\mathcal{F}(\mathcal{U})$ that could have been used to generate the data:
$$\mathcal{F}(\mathcal{U}) = \{ \mathcal{P} : (n_1, ... n_k) = G(\mathcal{P}, \textbf{u} | \pi ), \exists \pi \in \mathcal{S}_{\{1, .... , k\}} \}. $$
In the case of \textit{Dirichlet-DSM},  $\mathcal{F}(\mathcal{U})$ has a very simple expression. $\mathcal{F}(\mathcal{U})$ is directly equal to our previously described polytopes with Dirichlet vertices, which we have been calling $\Delta(Z)$. Furthermore, as we can see from our definition of $\Delta(Z)$, there is no dependence on the permutation $\pi$ that was used to generate the data. The ability to decouple class randomization from inference is one of the strengths of the method.
We now introduce how these random sets yield valid DS upper and lower probabilities. Consider a continuous test statistic $T(\hat{\mathcal{P}}, \mathcal{P}_0)$ which defines a measurable map $T(\cdot, \mathcal{P}_0):(\Delta^{k-1},\mathcal{M}_{\Delta ^{k-1}}) \rightarrow (\Omega, \mathcal{M}_{\Omega})$, where $\mathcal{M}_{\Delta}^{k-1}$ is the usual Borel algebra on the $k-1$ dimensional simplex and $\Omega$ is a one-dimensional measure space, most commonly $\mathbb{R}$ or $\mathbb{N}$. Given a measurable posterior random set, $\Delta(Z)$, via continuity of $T$, $T(\Delta(Z))$ as well as $\sup T(\Delta(Z))$ and $\inf T(\Delta(Z))$ are all continuous measurable sets in $\mathcal{M}_{\Omega}$. To perform a hypothesis test, the user now provides a measurable set $\Sigma \subset \mathcal{M}_{\Omega}$ that corresponds to the hypothesis of interest. \\\\
For the previous example in Section \ref{sec:hypo}, we tested if the parameters were uniform, $(p_1 = p_2 = p_3 = 1/3)$ with chi-squared test statistic $T(\mathcal{P}, \mathcal{P}_0) = \sqrt{\frac{\sum_{i=1}^3 (np_i - n/3)^2}{n/3}}$. With the test of uniformity, our corresponding set of interest is: $\Sigma = \{ T(\mathcal{P}, \mathcal{P}_0) \leq T(\mathcal{\hat{P}}, \mathcal{P}_0)_{\alpha\%} \}$ where $T(\mathcal{\hat{P}}, \mathcal{P}_0)_{\alpha\%}$ is the $\alpha$ quantile of $T(\mathcal{\hat{P}}, \mathcal{P}_0)$. \\\\
DS theory presumes the existence of a mass function $m$ such that:
$$m:2^\Omega \rightarrow [0,1] $$
such that:
$$m(\emptyset) = 0 $$
and
$$\sum_{S \in 2^{\Omega}} m(S) = 1.$$
Any probability measure can be written as a mass function by letting every measurable set be equal to its probability and setting unmeasurable sets to 0:
$$m(A) = \begin{cases}
P(A) & \text{ if } A \in \mathcal{M}_{\Omega} \\
 0 & \text{ if } A \not \in \mathcal{M}_{\Omega}
\end{cases}. $$
Moreover, for any set $S \subset 2^{\Omega}$ (as opposed to any measurable set), the belief of a set S is defined as the mass of sets which are contained within S: 
$$bel(S) = \sum_{A|A \subset S} m(A).$$
The plausibility of S is defined as the mass of sets that intersect S:
$$pl(S) = \sum_{A|A \cap S \neq \emptyset}  m(A) .$$
Due to this construction:
$$bel(A) \leq P(A) \leq pl(A).$$
With the above definitions, we can now  show that $\pi_{upper}$ and $\pi_{lower}$ are valid belief and plausibility functions for the one-sided hypothesis test previously described:
\begin{theorem}
When $T$ defines a distance metric, $\pi_{upper}$ and $\pi_{lower}$ define valid belief and plausibility functions for $S$ where $S= [T(\mathcal{\hat{P}}, \mathcal{P} )_{\alpha \%}  \infty)$.
\end{theorem}

\begin{proof}
First, since T defines a distance metric, if for a given $Z$:
$$T_{upper}(Z)  \leq T(\mathcal{\hat{P}}, \mathcal{P} )_{\alpha \%}  $$
this implies that:
$$ S=[T_{upper}(Z)_{\alpha } , \infty) \subset [T_{upper}(Z) , \infty).$$
Now, if we let $m$ be the counting measure:
\begin{align*}
    bel(S) &= \sum_{ [T_{upper}(Z), \infty) | [T_{upper}(Z), \infty) \subset S} m([T_{upper}(Z), \infty)) \\
    &= \sum_{T_{upper}(Z) \leq T(\mathcal{P}, \hat{\mathcal{P}}_{\alpha \%})} m([T_{upper}(Z), \infty))\\
    &= P(T_{upper} \leq T(\mathcal{P}, \mathcal{\hat{P}})_{\alpha \%,})\\
    &= \pi_{upper} 
\end{align*}
Likewise using the same argument, we can show that $pl(S) = \pi_{lower}$.
\end{proof}
Based on this theorem, we can view the belief function as the probability we observe a posterior set where \textit{every} element have a $T$ value less than the $\alpha \%$ cut-off. This belief function is what we have been referring to as our lower p-value. On the other hand, the plausibility function is the probability that we observe a posterior set where there exists \textit{at least} one element whose value $T$ is less than $\alpha$.
\subsection{Fiducial Interpretation}
In terms of Fiducial Inference, we will be showing that the previously defined "feasible" set $\Delta(Z)$ defines what is known as a Generalized Fiducial Quantity. Recall that given fixed latent parameters $z$ from a Dirichlet distribution, our inference procedure results in a polytope $\Delta(z)$ of parameter values for which there exist uniform random variables $(U_1, ... U_k)$ that could have generated the observed data. As above, we write the set of parameters that could have feasibly generated $(n_1, ..., n_k)$ as:
$$\mathcal{F}(\mathcal{U}) = \{ \mathcal{P} \in \Delta^k : (n_1, ... n_k) = G(\mathcal{P}, \textbf{u} | \pi ), \exists \pi \in \mathcal{S}_{\{1, .... , k\}} \}. $$
However, unlike DS, which directly works with the random posterior sets, $\mathcal{F}(\mathcal{U})$, in the Generalized Fiducial Scheme, when $\mathcal{F}(\mathcal{U})$ contains multiple elements, one selects an element according to some possible random rule V. Mathematically, we say that for any measurable set $S \in \Delta^k$ of feasible parameters, one selects a possibly random element $V(S)$ with support $\overline{S}$ where $\overline{S}$ in the closure of S. In our case, our random set is $\mathcal{F}(\mathcal{U})$  from which we choose the supremum, infimum, or average value of the polytope. As supremum, infimum, and average values are all random rules that lie within the closure, our upper, lower, and mean test distributions define valid Fiducial Generalized Quantities.

\section{The Advantages of DS Multinomial Hypothesis Testing}
We would now like to pivot to demonstrating features of our DS hypothesis testing procedure and the powerful properties though an investigation of DS's robust insight into testing.

\subsection{The Unknown class can help control Type I/Type II Error Rates}
Our DS inference procedure indicates if one should reject the null with confidence, fail to reject the null with confidence, or if we lack the information to conclude either way. If one restricts themselves to only the samples for which DS has given a certain conclusion, we can demonstrate that this extra confidence translates into a lower overall Type I/ Type II misclassification rate. To demonstrate, we perform a 4 dimensional test:
\begin{align*}
    H_0: & p = (0.25,0.25,0.25,0.25) \\
    H_1: & p \neq (0.25,0.25,0.25,0.25)
\end{align*}
where $(X_1, X_2, X_3, X_4) \sim Multinomial(n, (0.3,0.3,0.3,0.1))$ and $n$ is varied from 10 to 100. For each n, 250 samples are drawn and both a DS and a permutation-based chi-squared test are performed. From the chi-squared test, we will calculate the percent of tests that are correctly rejected, (Freq Percent Correct). And from the DS we will compute 2 values, first, the percent of total tests we correctly reject (DS Total Correct) and the percent of \textit{certain} tests we correctly reject (DS Certain Correct). From the results of the Figure \ref{fig:Certain1}, we can see that DS Certain Correct has the higher power, followed by DS Total Correct and then Freq Percent Correct. Moreover, using DS Certain Correct only comes with a small sacrifice in sample size, as only between 0.05-0.1 of the samples were determined to be uncertain with the number of uncertain tests intuitively going to 0 as sample size increase.
\begin{figure}[h]
    \centering
    \includegraphics[width = \textwidth]{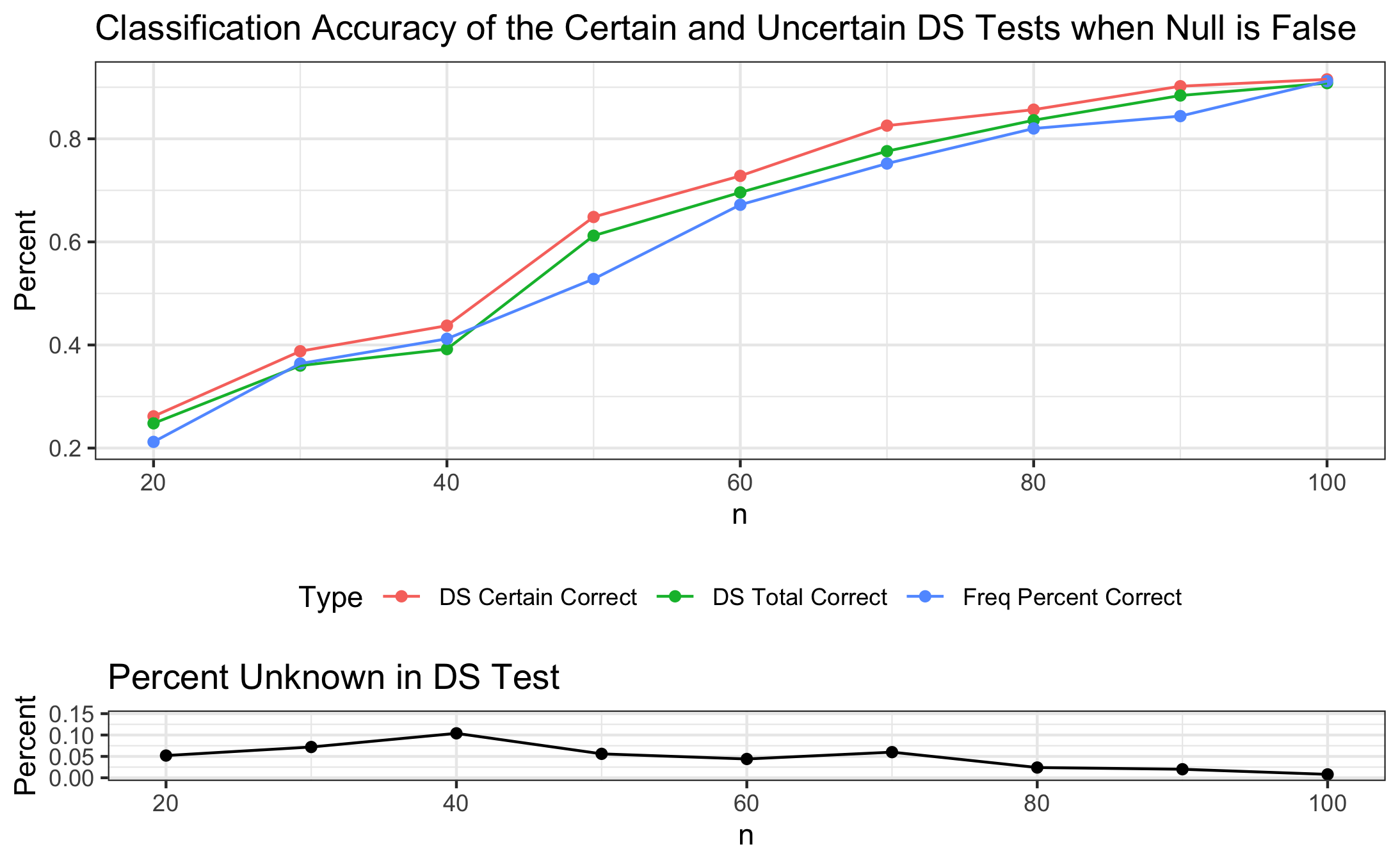}
    \label{fig:Certain1}
\end{figure}
To confirm that this increase in accuracy is not coming at the expense of preserving the 0.05 level, we perform the same simulation as before except this time the null is true. The results of this null simulation in Figure \ref{fig:Certain2} show that DS Certain Correct does indeed preserve a similar level to the frequentist test with the slightly lower level of DS Total Correct explainable by the uncertain classes. 

\begin{figure}[h]
    \centering
    \includegraphics[width = \textwidth]{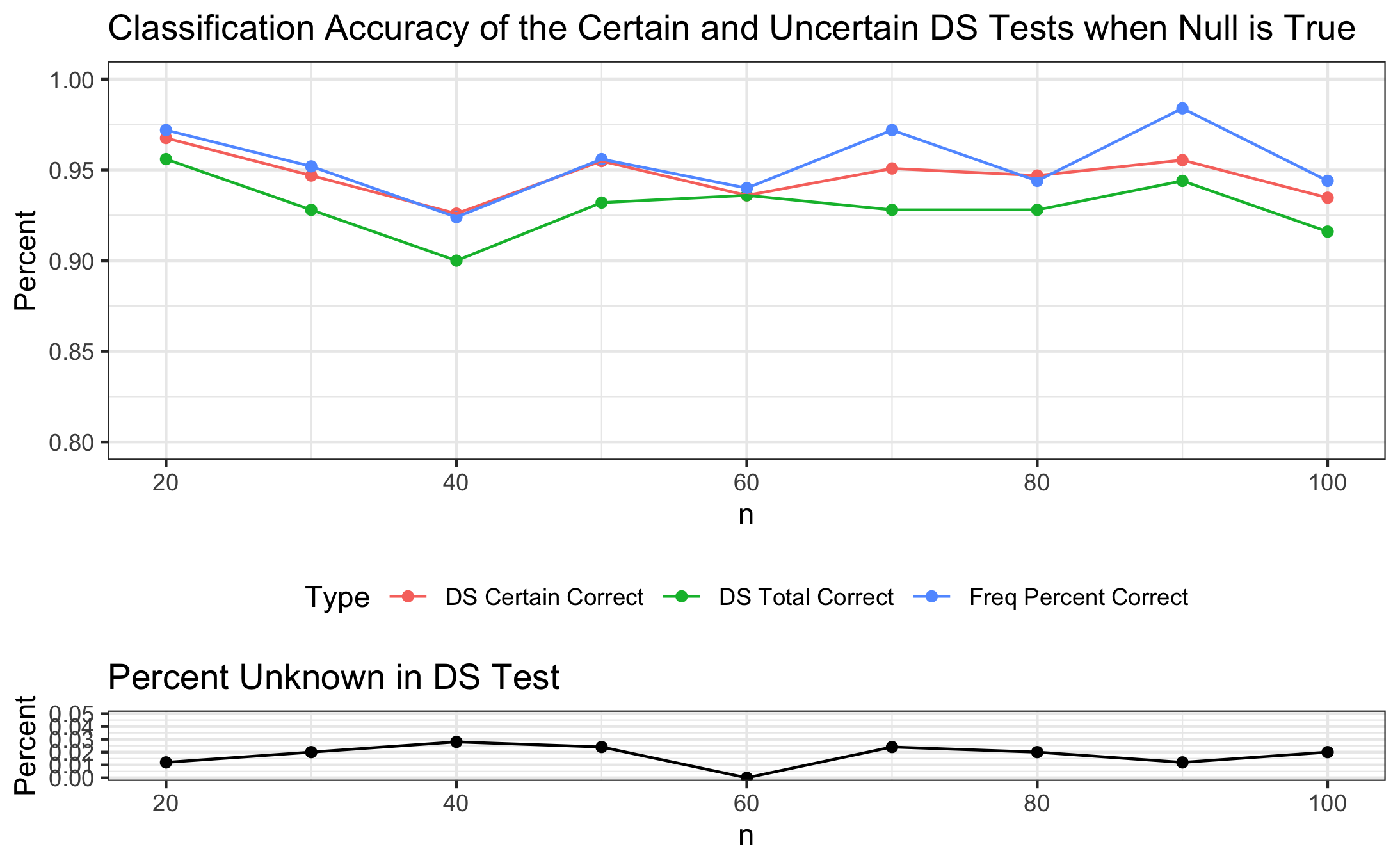}
    \label{fig:Certain2}
\end{figure}

\subsection{The Unknown Class gives insight into the difficulty of a test }
DS has a long history in the engineering community providing a way to comprehend the level of epistemic uncertainty in a hypothesis test through the inclusion of an ``unknown" option. In figure \ref{fig:Sample_Power}, we demonstrate the utility of this unknown class through a simulated 4-class Multinomial goodness-of-fit test. The alternate hypothesis had probabilities $(\frac{2}{6}, \frac{1}{6}, \frac{2}{6}, \frac{1}{6})$.  The test statistic used was the chi-square statistic and the sample size, $n$, ranged from 4 to 256. At each sample size, 1000 tests were performed and the upper and lower p-values were generated using 100 polytopes. The null distribution of the Frequentist test was simulated from 1000 draws from a multinomial distribution with equal probabilities. No multiple testing correction was performed. \\\\
In the first plot in Figure \ref{fig:Sample_Power}, we can see that while both tests reach the correct level of 0.05 as the sample size increases, they have very different behavior at low sample sizes. In the frequentist test, when the dimension-to-sample size ratio is one, the frequentist test rejects too few tests. This inability to reject is a well-documented high-dimensional phenomenon \cite{10.1214/18-AOAS1155SF} not just for multinomials, but for a large class of hypothesis tests \cite{wainwright_2019}. However, in the DS test, when the sample size is too small, there is a substantial possibility that the hypothesis test will return with an ``Uncertain" conclusion. Unlike the frequentist test which fails to reject nearly all, this unknown class is giving additional information to the user on the difficulty of making a conclusion at the current sample size. We see this additional information at work in Figure \ref{fig:Sample_PowerH1}. Once again, at low sample sizes, the frequentist test does not indicate if we are failing to reject the null because the null is true or because we lack the sample size to make any conclusions without performing a seperate power calculation. However, if we look at the DS test, the large probability of a test returning an ``Uncertain" result indicates that our issue is in the sample size. 
\begin{figure}
    \centering
    \includegraphics[width = 0.6\textwidth]{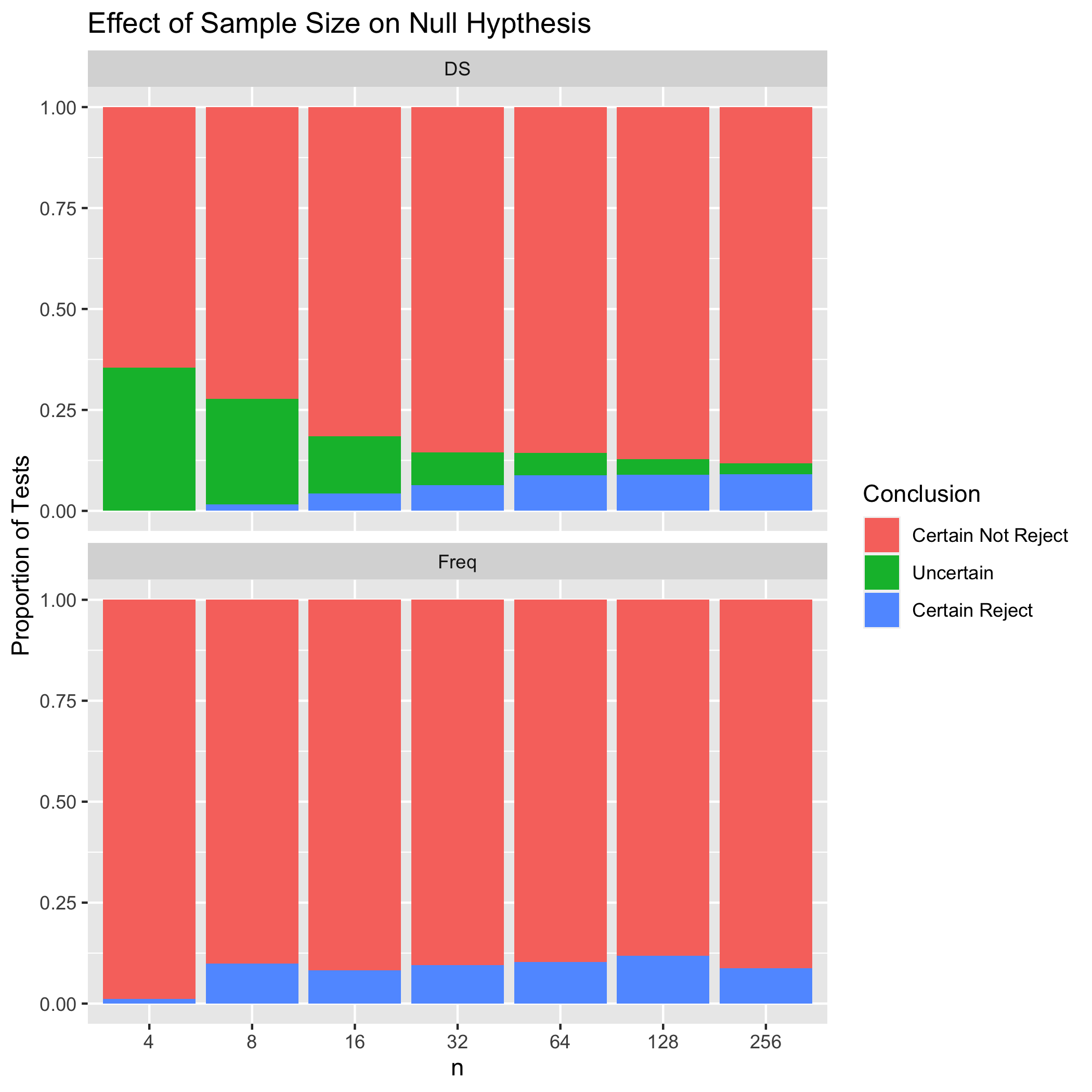}
    \caption[Power when null is true]{Results of Frequentist and DS hypothesis test of uniformity when the null is true. At low sample sizes, the Frequentist test gives a deceptively low rejection rate while the DS properly indicates the difficulty of this problem by having a large probability for an ``Uncertain" result.  }
    \label{fig:Sample_Power}
\end{figure}

\begin{figure}
    \centering
    \includegraphics[width =  0.6\textwidth]{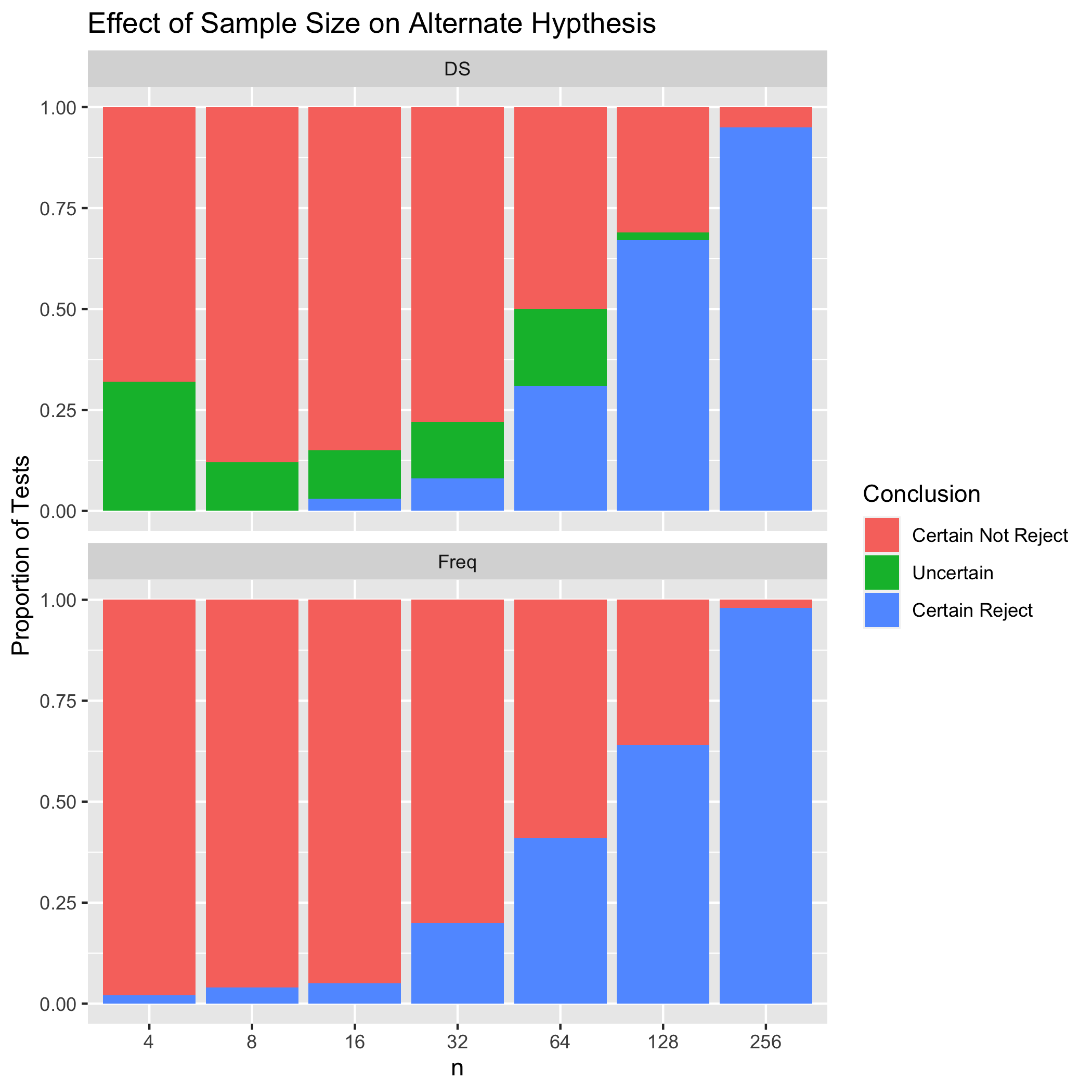}
    \caption[Power when null is False]{Results of Frequentist and DS hypothesis test of uniformity when the null is false. Once again, with the frequentist test, at low sample sizes, it is unclear if the alternate is correct or if the sample size is lacking, while the DS makes this clear through the unknown class that the sample size is lacking.}
    \label{fig:Sample_PowerH1}
\end{figure}

\subsection{DS can model Adversarial Attacks on a Test}
In addition to providing a novel characterization of power, our approach can also employ a technique called ``weakening" to evaluate the effect an adversarial attack would have on a hypothesis. Weakening is a common practice in reliability engineering when one gives additional weight for the unknown class either to satisfy long-running frequentist properties \cite{10.1214/10-STS322} or to account for potential uncertainties in the observed data \cite{33f451a631da44f9873b3e679d527013}. For our DS model, we will perform weakening by artificially increasing the width of our polytopes. Formally:
\begin{definition}[$\alpha$-Weakened Hypothesis Test] An $\alpha$-Weakened Hypothesis test is the same as the original DS hypothesis testing procedure except the Dirichlet parameters are drawn from:
$$(Z_0, ... , Z_k) \sim Dirichlet(1 + \alpha, x_1, ... , x_k) $$
where $\alpha > 0 $.
\end{definition}
Recalling that $Z_0$ directly controls the width of the random polytopes that serves as our posterior estimates, one can see how this directly increases the gap between the upper and lower test statistics. Since our DS testing procedure returns an ``unknown" when the upper and lower test statistics disagree, this weakening directly leads to increasing the possibility of a test returning an `unknown" result.
\\\\
As a demonstration of this, in Figure \ref{fig:weaken1}, we added weakening to a 4-class multinomial hypothesis test with the same alternate hypothesis as the previous simulation. The sample size is 128 and once again, no multiple testing correction was performed. From the results in the figure, we can see that increasing the amount of weakening leads to a dramatic increase in the probability of an uncertain result with a weakening of around 10-20 leading to almost all tests resulting in an unknown result.
\begin{figure}
    \centering
    \includegraphics[width = 0.6\textwidth]{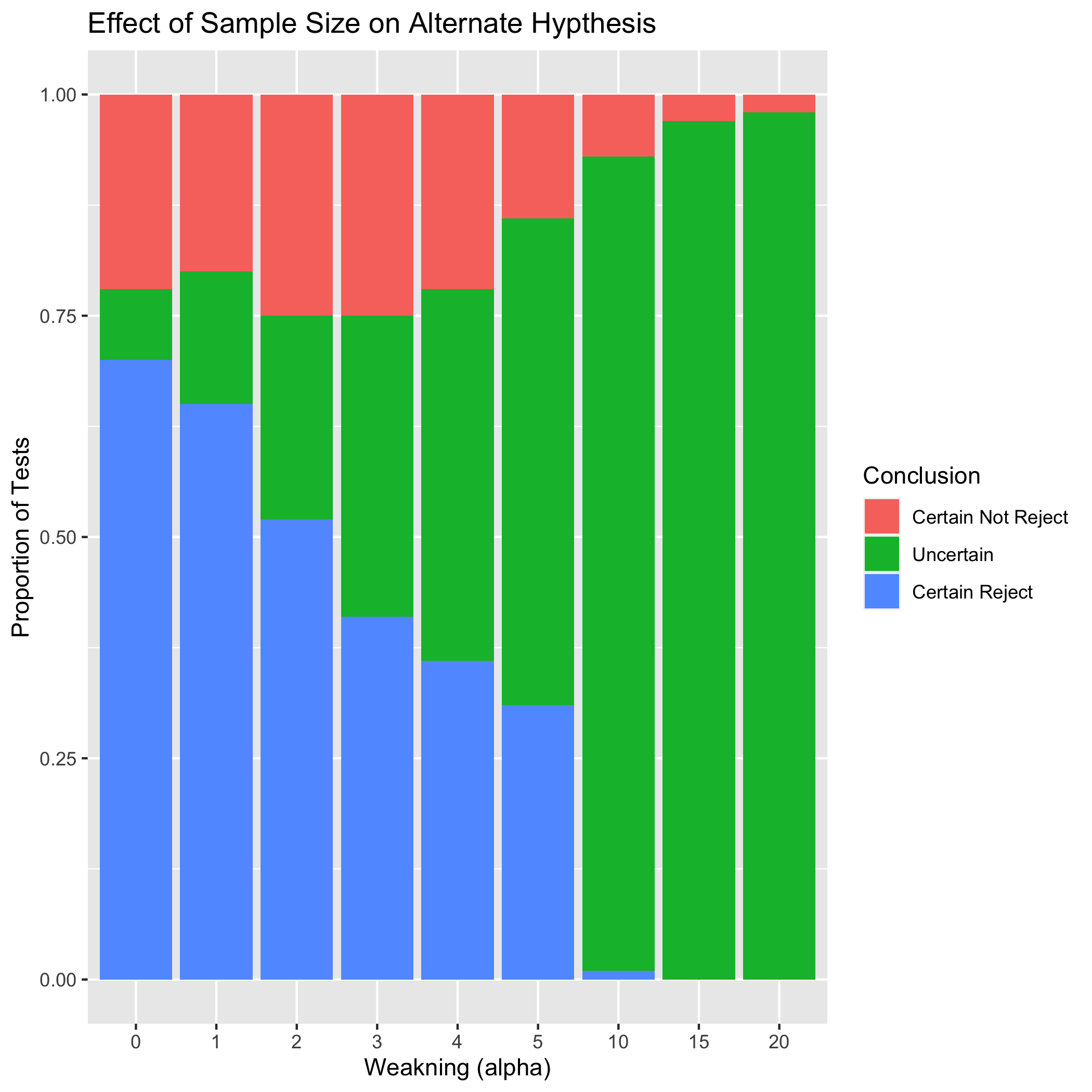}
    \caption[Effect of Weakening on Power]{Effect of weakening on a DS test of uniformity when the null is false. As the number of adversarial samples (represented by alpha) increases, the conclusions become increasingly more muddled. }
    \label{fig:weaken1}
\end{figure}
\\\\
In addition to allowing for the addition of user-specified uncertainty, this form of weakening comes with an easy-to-understand interpretation. For demonstration purposes, let us consider a 3-dimensional hypothesis test in which the observed counts were $(n_1,n_2,n_3)$ with total count $n^* = n_1 + n_2 + n_3$. Since:
$$(Z_0, Z_1, Z_2, Z_3) \sim Dirichlet(1,n_1,n_2,n_3) $$
component wise, the $Z_i$'s have the distributions:
\begin{itemize}
    \item $Z_0 \sim Beta(1, n + 1)$
    \item $Z_1 \sim Beta(n_1,  n+ 1 - n_1 )$
    \item $Z_2 \sim Beta(n_2,  n+ 1 - n_2)$
    \item $Z_3 \sim Beta(n_3, n+1 - n_3)$
\end{itemize}
As described previously, $Z_0$ is then added to each of the other Zs to define the edges of the polytope:
\begin{itemize}
    \item $(Z_1 + Z_0, Z_2, Z_3)$
    \item $(Z_1 , Z_2+ Z_0, Z_3)$
    \item $(Z_1 , Z_2, Z_3+ Z_0)$
\end{itemize}
So $Z_0$ controls the width of our polytope while $Z_1, Z_2, Z_3$ controls the location of the polytope. \\\\
Now, let us say that additional $\alpha$ data points has been added to our dataset. We are unclear which of the 3 classes it has been added to, so we act as though it was as likely to have been added to any of the 3 class. This would result in observed counts of $(n_1 + \frac{\alpha}{3}, n_2+ \frac{\alpha}{3}, n_3+ \frac{\alpha}{3})$ and $n = n_1 +n_2 +n_3 + \alpha$. Component wise, our $Z_i$'s will become:
\begin{itemize}
    \item $Z_0 \sim Beta(1 + \alpha, n + 1 + \alpha )$
    \item $Z_1 \sim Beta(n_1,  n+ 1 + \alpha - n_1 )$
    \item $Z_2 \sim Beta(n_2,  n+ 1 + \alpha - n_2)$
    \item $Z_3 \sim Beta(n_3, n+1 + \alpha - n_3)$
\end{itemize}
Thus, our width, $Z_0$, has been increased proportional to the addition of $\alpha$ datapoints. Therefore, \textit{$\alpha$- weakening is analogous to the effect of adding $\alpha$ datapoints to any of the  classes}. If we have previously confidently rejected or failed to reject before weakening, but then became unknown after $\alpha$-weakening, this indicates that $\alpha$ data points are sufficient to contradict our previous results. Combine this with the previous assertions that $\pi_{lower}$, $\pi_{upper}$ serve as a reasonable bound on the behavior of $\pi_{freq}$ and we can also make reasonable claims about the behavior of $\pi_{freq}$ under $\alpha$-weakening. \\\\
Returning to Figure \ref{fig:weaken1}, this simulation indicates that if 20 datapoints were added to the original 128, it is possible to have the hypothesis reject (hence why $\pi_{lower} \leq 0.05)$ and not reject (hence by $\pi_{upper} \geq 0.05$). Stated another way, with 20 additional datapoints, you can make a hypothesis test return whichever result you would like. 

\section{Using Dempster-Shafer Inference to study text data in Verbal Autopsies }

In this section, we demonstrate Dempster-Schafer inference in the context of verbal autopsies (VAs).  In particular, in low-resource settings, there is a lack of critical public health infrastructure to comprehensively record all births and deaths~\citep{mahapatra2007civil, setel2007scandal}. When deaths occur outside of hospitals or are not routinely recorded, VAs are used as a cost-effective approach to learning about the distribution of deaths by cause.  VAs consist of a survey with a surviving family member or caregiver.  The typical VA survey contains a mix of categorical (e.g. did the decedent have a fever?) and free response text describing symptoms that a subject seemed to exhibit before passing away. The symptoms are then tabulated and inputted into a specialized designed classification algorithm such as Tariff 2.0 \cite{Serina2015}, InterVA \cite{pmid31146736}, and InsilicoVA \cite{doi:10.1080/01621459.2016.1152191}. 

As collecting data for VA surveys involves training personnel in countries that may already be struggling to ensure adequate public health resources and interviews may bring back potentially traumatic memories, there have been recent developments in choosing more efficiently the types of questions to ask during the interview \cite{yoshida2023bayesian}. Building off of this, we would like to use our DS test to help determine which words in the free-response section may help us discriminate various causes of death.

We take data from the Population Health Metrics Research Consortium (PHMRC) which contains verbal autopsy and physical autopsy information from 7841 adults taken from 6 sites in developing nations around the world \cite{Murray2011}. As this contains information about verbal autopsy as well as the true cause of death it is often considered a ``Gold Standard" dataset that verbal autopsy algorithms are validated on. In this example, we will focus our attention on ability of the text portion of PHMRC to discriminate between potentially similar causes of death, which in this case are the 8 different types of cancers: Breast, Cervical, Colorectal, Esophageal, Leukemmia/Lymphomas, Lung, Prostate, and Stomach. 854 deaths were from one of these cancers and had the top ``GS Level 1"  in the cause of death assessment. To determine these causes of death, there were a total of 624 uniquely identified words that were used at least once. \\
We will be testing the DS multinomial test as a goodness-of-fit hypothesis that the probability of the usage of a word across causes of death matches the prevalence of the diseases. Words that reject the hypothesis indicate a potential utility in discerning causes of death. For each word,  a DS Hypothesis test with 1000 replicates was performed. As a point of comparison, a resampled chi-squared test with 1000 replicates was also performed. Both test's conclusions were made at the $\alpha =0.05$ level. 

\begin{figure}[ht]
    \centering
    \includegraphics[width = \textwidth]{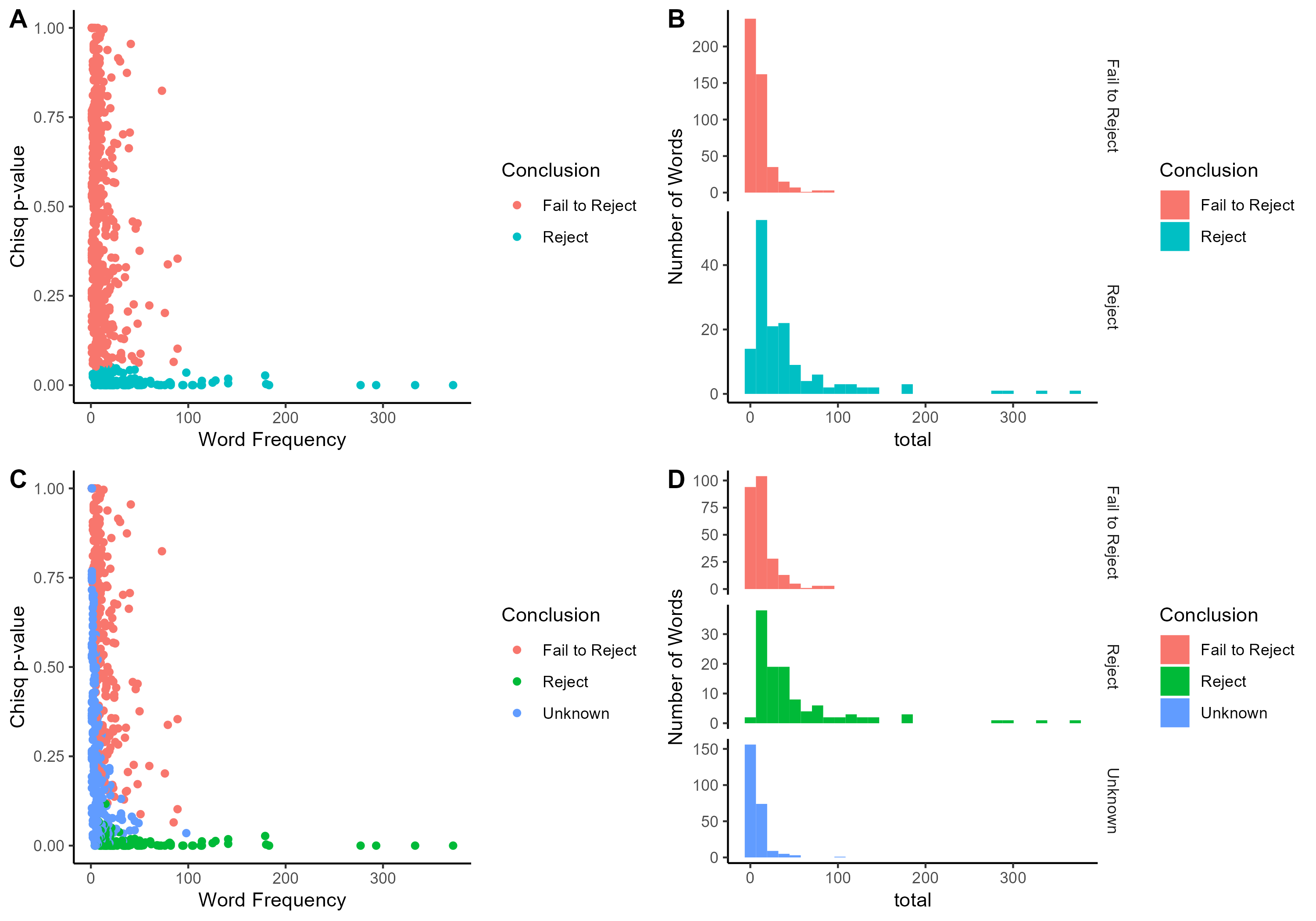}
    \caption{(Top Left): Chi squared test of word significance. The vertical grey line represents the sample size of 50 which was used in \cite{Serina2015}. (Top Right): Histograms of the sample size of the words by Hypothesis test conclusion. (Bottom Left): DS Hypothesis test of word significance.}
    \label{fig:PHMRC_pvalue}
\end{figure}

Figure \ref{fig:PHMRC_pvalue} A and C demonstrate the number of people who used a word as well as it's goodness-of-fit p-value under the chi-squared test and the DS test. We can see in both that the distribution of word count is heavily right skewed, This indicates that many words were used only by a small number of subjects and accords with well known, empirical observations of  word use in natural language \cite{zipf1999psycho}.

In both tests, we see an agreement on the words for which we have evidence to reject the null hypothesis albeit. However, the two tests differ when it comes to the words that the chi-square test failed to reject. While the frequentist test categorizes anyone with a p-value above 0.05 in the ``Fail to Reject" cateogry, the DS test splits this group into two similarly sized groups. These are the words for which Fail to Reject with confidence, and words for which we lack evidence to conclude either way.  The words for which we lack the evidence to conclude either way includes some medical terms such as "pneumonia" (used 31 times) and "diabet[us] (used 26 times)". While pneumonia and diabetus are known to be associated with various forms of cancer or cancer treatment \cite{pmid31516903,Giovannucci2010}, the low sample size could indicate an unfamiliarity with medical terminology in the low-resource areas where VAs were performed. On the other hand, words such as "vomit" (used 28 times) and "urin[e]" (used 23 times) were rejected with confidence by the DS test. In total, by being able to separate words which maybe helpful to discriminate between causes of death but may not see sufficient usage due to educational or other resource related issue from words for which we have some confidence are unrelated to the cause of death gives the investigator a deeper insight into the data, as well as gives potential inspiration for designing more focused interviews in the future. 

Finally, we demonstrate the DS tests' results as a variable prescreening procedure. Words that were deemed significant were entered into a purpose-built classification method for verbal autopsies called Tarrif \cite{Serina2015}. On table \ref{tab:1} , we see that the Tariff algorithm, when trained on all 624 words, is able to classify each verbal autospy into one of the 8 different types of cancer with an accuracy of 31.5\%. This compares favorably to the meager accuracy of 15. 97\% according to the rule of thumb in the Tariff to use all words used at least 50 times. This poor performance is corroborated by the results of the previous goodness-of-fit tests, as many words which both the chi-squared and DS test deemed significant were used by less than 50 people. This seems to indicate that rule of thumb to ignore any words less than 50 times which was presented in  \cite{Serina2015} is sub-optimal. Instead, we recommend that a word's inclusion should be a factor of both sample size and the strength of the evidence, which are factors that are considered in the chi-squared and the DS test. Finally, the best preforming models were those trained on the rejected words by the chi-squared an DS test, albeit with the DS test gave the same classification accuracy based on 114 words rather than the 149 of the chi-squared test. The fact that fewer words were rejected by DS can be explained by how the DS procedure is classifying each word under a trivariate rather than binary scheme, and thus it is expected that some borderline words which were rejected by chi-squared maybe classified as unknown by DS. 

\begin{table}[h]
\centering
\begin{tabular}{|c|c|c|}
\hline
Scheme  &Number of Words & Accuracy \\ \hline
    Use all Words  & 624 & 31.25 \%\\ \hline
    Use all Words that were used atleast 50 times & 60 & 15.97 \% \\ \hline
    Use all words there were rejected by chisq at $\alpha =0.05$ & 149 & \textbf{40.28} \% \\\hline
    Use all words there were rejected by DS at $\alpha =0.05$ & 114 & \textbf{40.28} \% \\ \hline
    Use all words there were rejected or Unknown by DS at $\alpha =0.05$ & 362 & {35.42  } \% \\ \hline
\end{tabular}
\label{tab:1}
\caption{Classification Accuracy of the Tariff 2.0 algorithm given various  }
\end{table}

\section{Conclusion}
In this paper, we have demonstrated how to construct an alternative approach to p-values for multinomial data that allows a test to return an "unknown" result. By allowing for this third category, we believe that this allows for one to gain new insights into the data while simultaneously keeping much of the intuition and implications that make p-value based hypothesis testing appealing in the first place. It is the authors belief that such an approach will provide an important voice in the continued conversation in the future use of p-values in scientific endevors.

\renewcommand*\rot[2]{\multicolumn{1}{R{#1}{#2}}}

\bibliography{biblio}

\end{document}